\documentclass[reprint,pra,amsmath,amssymb, twocolumn,superscriptaddress, raggedbottom]{revtex4-2}

\usepackage{graphicx}
\usepackage{dcolumn}   
\usepackage{bm}        
\usepackage{bbm}
\usepackage[colorlinks=true, urlcolor=blue,citecolor=blue,anchorcolor=blue,breaklinks=true]{hyperref}
\usepackage[capitalise]{cleveref}
\usepackage{dsfont}
\usepackage{mathtools}
\usepackage{amsthm, thmtools}
\usepackage{amsmath}
\usepackage{empheq}

\newcommand{\microspace}{\mspace{0.5mu}}

\declaretheorem[
]{theorem}
\declaretheorem[
]{lemma}
\declaretheorem[
]{corollary}

\def\<{\langle}
\def\>{\rangle}
\def \lket {\left|}
\def \rket {\right\rangle}
\def \lbra {\left\langle}
\def \rbra {\right|}
\newcommand{\ket}[1]{\lket\microspace #1 \microspace\rket}

\newcommand{\bra}[1]{\lbra\microspace #1 \microspace\rbra}

\newcommand{\Tr}{\text{Tr}}
\usepackage{mathtools}

\DeclarePairedDelimiter{\ceil}{\lceil}{\rceil}

\usepackage{bm}
\let\vec\bm 

\begin{document}
\title{Simultaneous Stoquasticity}

\author{Jacob~Bringewatt}
\email{jbringew@umd.edu}
\affiliation{Joint Center for Quantum Information and Computer Science, NIST/University of Maryland, College Park, Maryland 20742, USA}
\affiliation{Joint Quantum Institute, NIST/University of Maryland, College Park, Maryland 20742, USA}

\author{Lucas~T.~Brady}
\email{lucas.t.brady@nasa.gov}
\affiliation{Joint Center for Quantum Information and Computer Science, NIST/University of Maryland, College Park, Maryland 20742, USA}
\affiliation{Joint Quantum Institute, NIST/University of Maryland, College Park, Maryland 20742, USA}
\affiliation{Quantum Artificial Intelligence Laboratory, NASA Ames Research Center, Moffett Field, California 94035, USA}
\affiliation{KBR, 601 Jefferson St., Houston, TX 77002, USA}

\date{\today}
\begin{abstract}
Stoquastic Hamiltonians play a role in the computational complexity of the local Hamiltonian problem as well as the study of classical simulability. In particular, stoquastic Hamiltonians can be straightforwardly simulated using Monte Carlo techniques.  
We address the question of whether two or more Hamiltonians may be made simultaneously stoquastic via a unitary transformation. 
This question has important implications for the complexity of simulating quantum annealing where quantum advantage is related to the stoquasticity of the Hamiltonians involved in the anneal. 
We find that for almost all problems no such unitary exists and show that the problem of determining the existence of such a unitary is equivalent to identifying if there is a solution to a system of polynomial (in)equalities in the matrix elements of the initial and transformed Hamiltonians. Solving such a system of equations is NP-hard. We highlight a geometric understanding of this problem in terms of a collection of generalized Bloch vectors.
\end{abstract}

\maketitle

\section{Introduction} The efficient simulation of quantum phenomena is essential to understanding chemistry, materials, and physics, and similarly the lack of efficient classical simulation is critical to the long-term applicability of quantum computing. One of the key properties that can make a Hamiltonian easy to simulate classically is stoquasticity \cite{Bravyi2008}, a basis dependent property where the off-diagonal matrix elements are real and non-positive \footnote{In the mathematics literature, such matrices are called Z-matrices or negative Metzler matrices.}. Such stoquastic Hamiltonians do not suffer from the sign problem allowing classical simulation of their ground state properties via Monte Carlo techniques \cite{barker1979montecarlo,Crosson2021stoq}.

Stoquastic Hamiltonians have been especially important in the development of quantum annealing \cite{kadowaki1998} and quantum adiabatic computation \cite{farhi2000quantum}. Adiabatic quantum computing is quantum universal \cite{Aharonov2007}, but the proof relies on non-stoquastic Hamiltonians. There is growing evidence that adiabatic computing with stoquastic Hamiltonians is no more powerful than classical computing \cite{Brady2016,Crosson2016,Jiang2017,Crosson2021stoq,Bringewatt2020} except in contrived highly non-local settings \cite{Hastings2020}. In complexity theory, stoquastic Hamiltonians appear in the definition of the complexity class \texttt{StoqMA}, which characterizes the computational hardness of the local Hamiltonian problem for stoquastic Hamiltonians \cite{bravyi2006merlinarthur}.

A large body of literature has been built up around the problem of finding stoquastic \cite{Li_2015,hen_resolution_2019} or nearly stoquastic \cite{wan2020mitigating,hangleiter_easing_2020} bases for Hamiltonians. The corresponding unitary basis change is said to ``cure'' the non-stoquastic Hamiltonian. While the existence of such a basis is guaranteed by the diagonalizability of Hermitian matrices (the locality of such a basis is not guaranteed), finding such a basis change is an NP-hard problem \cite{troyer_computational_2005,marvian_computational_2019,Klassen2019}. However, this literature has mostly focused on just curing a single Hamiltonian's sign problem. In order to run simulated quantum annealing \cite{hastings2013obstructions,Crosson2016} or otherwise simulate the behavior of adiabatic computation, both annealing Hamiltonians must be stoquastic. This raises the question not just of how to find a basis in which two Hamiltonians are simultaneously stoquastic but further whether such a basis even exists. Our work along this direction is complementary to the results in Ref.~\cite{marvian_computational_2019} where the authors consider the problem of stoquasticizing a local Hamiltonian consisting of a sum of local terms and showed that this Hamiltonian can be stoquasticized if and only if all terms can be simultaneously stoquasticized. Furthermore, they showed that it is NP-hard to find a basis that accomplishes this. 

To formally state our problem of interest: let $\texttt{Stoq}$ be the set of all stoquastic matrices. Given a set of Hamiltonians $S=\{H_1, H_2, \cdots H_m\}$ defined on a $d$-dimensional Hilbert space $\mathcal{H}_d$, does there exist a single unitary $U$ that simultaneously cures the non-stoquasticity of (``stoquasticizes'') all $H_j\in S$; that is, $\exists$ $U$ such that $UH_jU^\dagger\in \texttt{Stoq}$ for all $j$?

Using the mathematical theory of simultaneous unitary similarities ~\cite{specht1940theorie, wiegmann1961necessary,gerasimova2013simultaneous, jing2015unitary}, we find that  the problem reduces to determining if there exists a solution to a system of polynomial (in)equalities. For $m>2$ and/or $d>2$, the resulting system of polynomial equations does not generically have a solution, and therefore a simultaneously stoquasticizing unitary does not always exist. In fact, we show almost every set $S$ of Hamiltonians is not simultaneously stoquastizable. By considering a generalized Bloch vector representation, we can geometrically interpret our results, connecting to the literature on the geometry of quantum states \cite{harriman1978geometry, hioe1981nlevel, byrd2003characterization, kimura2003bloch, kimura2005bloch, bertlmann2008bloch}.

This result has broad implications for adiabatic quantum computing, where annealing between two Hamiltonians that are not simultaneously stoquasticizable should be hard to simulate classically, independent of basis. 
The more general theory of simultaneous transformation of two or more Hermitian operators plays a key role in other areas of quantum physics, the most obvious being simultaneous diagonalizability governing the commutativity and compatability of observables. Similarly simultaneous unitary congruence has been used to show that quantum seperability is connected to the simultaneous hollowability of matrices \cite{neven2018quantum}, and simultaneous orthogonal equivalence connects to local unitary equivalence of a pair of quantum states \cite{jing2016}.

\section{Lie algebras and Lie groups}Formally, any Hamiltonian $H\in\mathbb{C}^{d\times d}$ is (up to a physically irrelevant shift by a multiple of identity) an element of the Lie algebra $\mathfrak{su}(d)$. Here, we take the usual physicist convention that $\mathfrak{su}(d)$ consists of the set of all $d\times d$ traceless, Hermitian matrices. This is known as the fundamental representation. The Lie algebra $\mathfrak{su}(d)$ has real dimension $d^2-1$ and, therefore, any element of the Lie algebra may be expanded in a basis of $d^2-1$ elements of the algebra, which we choose to obey the standard orthonormality condition
\begin{align}\label{eq:normalization_su}
    \mathrm{Tr}(\hat{\lambda}_i\hat{\lambda}_j)=2\delta_{ij}.
\end{align}
We also have that
\begin{align}\label{eq:basisproduct}
    \hat{\lambda}_i\hat{\lambda}_j=\frac{2}{d}\delta_{ij} I+if_{ijk}\hat{\lambda}_k+d_{ijk}\hat{\lambda}_k,
\end{align}
where $I$ is the identity matrix and $f_{ijk}$ and $d_{ijk}$ are the totally anti-symmetric and symmetric structure constants, respectively.  We use the convention of summing over repeated indices. 
 A standard choice of basis is the generalized Gell-Mann basis. It is made up of $d(d-1)/2$ symmetric matrices,
\begin{subequations}
\begin{equation}\label{eq:gmx}
    \hat\lambda_{jk}^{(x)}=\ket{j}\bra{k}+\ket{k}\bra{j},\quad (1\leq j<k\leq d),
\end{equation}
$d(d-1)/2$ skew-symmetric matrices,
\begin{equation}\label{eq:gmy}
    \hat\lambda_{jk}^{(y)}=-i\ket{j}\bra{k}+i\ket{k}\bra{j},\quad (1\leq j<k\leq d),
\end{equation}
and $d-1$ diagonal matrices,
\begin{equation}\label{eq:gmd}
    \hat\lambda_{j}^{(\mathrm{diag})}=\sqrt{\frac{2}{j(j+1)}}\mathrm{diag}(\underbrace{1,\cdots, 1}_{j}, -j, 0, \cdots, 0),
\end{equation}
\end{subequations}
where in this final equation we have $j\in\{1,\cdots, d-1\}$. For $\mathfrak{su}(2)$, the generators defined in this way are the familiar Pauli operators, which motivates the $x,y$ superscipts for the symmetric and skew-symmetric generalized Gell-Mann matrices, respectively. For $\mathfrak{su}(3)$, they are the Gell-Mann matrices. 

We can write any traceless Hamiltonian $H$ in this basis as
\begin{align}\label{eq:blochform}
    H=\vec b\cdot\hat{\vec \lambda},
\end{align}
where $\hat{\vec\lambda}$ is a vector of basis elements and $\vec{b}\in\mathbb{R}^{d^2-1}$ is the so-called (generalized) \textit{Bloch vector} corresponding to $H$. We consider grouping the components of $\hat{\vec\lambda}$ into subsets matching the basis elements defined in Eqs.~(\ref{eq:gmx})-(\ref{eq:gmd}) as follows: let $\mathcal{X}$, $\mathcal{Y}$ and $\mathcal{D}$ be the sets of indices corresponding to the symmetric, skew-symmetric, and diagonal generalized Gell-Mann matrices, respectively. We have $\{\mathcal{X}, \mathcal{Y}, \mathcal{D}\}=\{1,\cdots, d^2-1\}$.

Therefore, there exists an isomorphism $S\cong B$ between a set $S=\{H_1, H_2, \cdots, H_m\}$ of traceless Hamiltonians and a set of corresponding Bloch vectors $B=\{\vec b^{(1)}, \vec b^{(2)}, \cdots \vec b^{(m)}\}$. These Bloch vectors will simplify a number of proofs and provide a valuable geometric interpretation of our results.

\section{Simultaneous stoquasticity}Given the set $S=\{H_1, H_2, \cdots, H_m\}$ we want to solve the decision problem: Does there exist a unitary $U$ such that $H_j'=UH_jU^\dagger\in\texttt{Stoq}$ for all $H_j'\in S':=\{UH_1U^\dagger, UH_2U^\dagger, \cdots, UH_mU^\dagger\}$? 

Observe that the choice of trace of the Hamiltonians and of the Hermitian generators of $U$ can be chosen to be zero with no physical consequence. Therefore, without loss of generality, we restrict our consideration to traceless $H_j\in S$ and to special unitaries $U\in SU(d)$. 

This assumption allows us to directly describe the problem in terms of Bloch vectors as detailed in the previous section. In particular, consider the sets of Bloch vectors $B\cong S$ and $B'\cong S'$. In the space of Bloch vectors $\texttt{Stoq}$ corresponds to the subset of Bloch vectors such that $b_j=0$ for $j\in\mathcal{Y}$ and $b_j\leq 0$ for $j\in\mathcal{X}$. The decision problem is now that of finding whether there exists a unitary $U$ such that the vectors $\vec b'\in B'$ all fall in this subspace. 

When $d=2$, the Bloch space is of dimension $d^2-1=3$ and $\texttt{Stoq}$ is easily visualizable as the $\pm \hat{\sigma}^{(z)}, -\hat{\sigma}^{(x)}$ half-plane, as depicted in Fig.~\ref{fig:bloch2d}. In this case, it is well-known that $SU(2)$ is a double-cover of $SO(3)$, and therefore we can visualize the action of unitaries on $S$ as rotations of the collection of vectors $B$. It is simple to observe that the answer to our decision problem is yes if and only if the vectors $B$ all lie in a single half-plane. This plane can then be rotated via some unitary to align with the $\texttt{Stoq}$ half-plane. This is always possible if $m\leq 2$ and $d=2$.

\begin{figure}
    \centering
    \includegraphics[width=0.5\columnwidth]{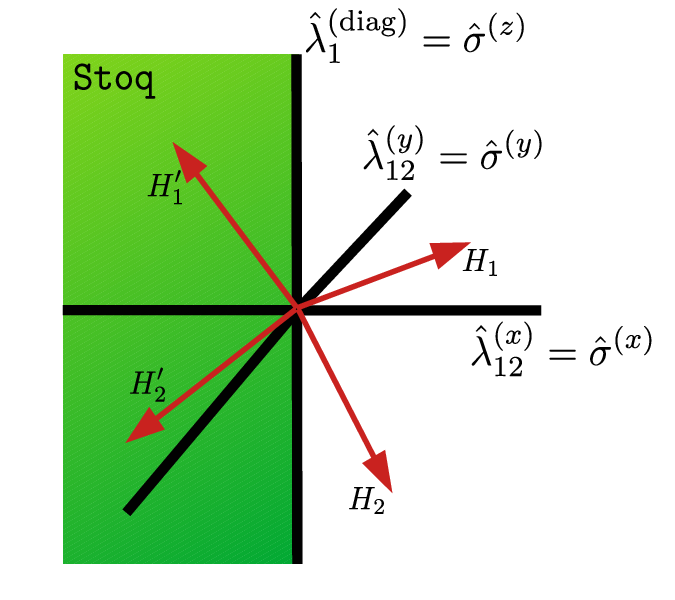}
    \caption{For a qubit, the geometric representation  of $\texttt{Stoq}$ in Bloch vector space is the $\pm \hat{\sigma}^{(z)}, -\hat{\sigma}^{(x)}$ half-plane. Observe that for two Hamiltonians $H_1, H_2$, represented by their Bloch vectors, there always exists a unitary which can simultaneously take both $H_1', H_2'$ to the $\texttt{Stoq}$ subspace. }
    \label{fig:bloch2d}
\end{figure}

We seek to generalize and formalize this geometric intuition for $d>2$. We will make use of the mathematical theory of simultaneous unitary similarities~\cite{specht1940theorie, wiegmann1961necessary, gerasimova2013simultaneous, horn2012matrix} and the related theory of simultaneous invariants \cite{procesi1976invariant}. Two ordered sets of $m$ matrices $S, S'$ are \emph{simultaneously unitarily similar} if there exists a unitary $U$ such that $H_j'= UH_jU^\dagger$ for all $H_j\in S$, $H_j'\in S'$.  In this terminology our goal is to determine if the set $S$ of Hamiltonians is simultaneously unitarily similar to a set $S'\in\texttt{Stoq}$. 

Define a \emph{word} on a set $T$ as any formal product of nonnegative powers of the elements $t_j\in T$. We then have the following theorem due to Ref.~\cite{gerasimova2013simultaneous}.
\begin{theorem}\label{thm:simsim}
The ordered sets of Hermitian matrices $S=\{H_1, \cdots, H_m\}$ and  $S'=\{H_1',\cdots, H_m'\}$ are simultaneously unitarily similar if and only if $\mathrm{Tr}[ w(S)]=\mathrm{Tr}[w(S')]$ for all words $w$ in $S,S'$.
\end{theorem}

The quantities $\mathrm{Tr}[w(S)]$ are known as \emph{trace invariants} under simultaneous unitary similarity. Unfortunately, Theorem~\ref{thm:simsim} is a practically useless condition since it requires the checking of all words in $S$ and $S'$. To get around this issue, one must demonstrate that only a finite set of independent words exist~\cite{paz1984application, pappacena1997upper, djokovic2007unitarily, horn2012matrix}. 

When checking the unitary similarity of a single pair of Hermitian matrices $H, H'$, it is often quoted in the physics literature (typically in the context of density matrices) that it is sufficient to check the equivalence of the trace invariants $\mathrm{Tr}[H^k]$ and $\mathrm{Tr}[H'^k]$ for $k\in[1,d]$ \cite{hioe1981nlevel, schlienz1996maximal, mahlerquantum, geometryquantum}. While perhaps intuitively obvious---as Hermitian matrices have $d$ real eigenvalues---this is typically stated without proof. For completeness, we give such a proof in the Appendix. 

More generally, for $m\geq 2$, we can show that it is sufficient to consider word lengths up to
\begin{equation}
    \label{eq:lmax}
    \ell_{\mathrm{max}}=\mathrm{min}\begin{cases}
    \ceil{\frac{(cd)^2+2}{3}}\\
    cd\sqrt{\frac{2(cd)^2}{cd-1}+\frac{1}{4}}+\frac{cd}{2}-2.
    \end{cases}
\end{equation}
where $c$ is the minimum integer such that $(c^2-3c+2)/2\geq m$. For instance, if $m=2$, $c=4$. The proof mostly follows Refs.~\cite{gerasimova2013simultaneous, paz1984application, pappacena1997upper}, and the derivation of this expression is demonstrated in the Appendix. 

Therefore, the decision problem of whether there exists a unitary that simultaneously stoquasticizes $S$ is equivalent to determining if there exists a solution to the system of polynomial (in)equalities in the matrix elements of $H'\in S'$:
\begin{subequations}
\begin{align}
    &\mathrm{Tr}[w(S)]=\mathrm{Tr}[w(S')], \quad &\forall \, |w|\leq \ell_{\mathrm{max}} \label{eq:polysys1}\\
    &\mathrm{Re}(H'_{jk})\leq 0, \quad &\forall j\neq k, H'\in S' \label{eq:polysys2}\\
    &\mathrm{Im}(H'_{jk})= 0, \quad &\forall j\neq k, H'\in S' \label{eq:polysys3}.
\end{align}
\end{subequations}

This amounts to a system of $\sum_{\ell=2}^{\ell_\mathrm{max}} m^\ell+md(d-1)/2\sim m^{O((cd)^{3/2})}$ polynomial equations and $md(d-1)/2$ inequality constraints on $md^2$ real variables. Many of the equations for different words in Eq.~(\ref{eq:polysys1}) will end up being redundant due to symmetries such as the cyclicity of the trace and algebraic dependence of the resulting trace invariants. 
Independent of if one can identify the minimal set of such constraints, solving the decision problem of whether or not a solution exists to this set of polynomial (in)equalitites is NP-hard and lies in PSPACE~\cite{canny1988some}. Therefore, identifying if $S$ is simultaneously stoquasticizable is completely intractable for large problem instances.

\section{A no-go result}
Given this computational difficulty, we also present the following no-go result. 
\begin{theorem}\label{thm:nogo}
A necessary condition for $S$ to be simultaneously stoquasticizable is that for every eigenvalue $\lambda\neq 0$ of $i[H_i,H_j]$ there is another eigenvalue $-\lambda$ of $i[H_i,H_j]$ (paired eigenvalue condition) for all $H_i\neq H_j\in S$. 
\end{theorem}
\begin{proof}
Any two matrices $H_i', H_j'\in\texttt{Stoq}$ have all real matrix elements. Therefore, the Hermitian matrix $C'=i[H_i', H_j']$ must be skew-symmetric. Skew-symmetric matrices have the property that all eigenvalues are paired. As eigenvalues remain unchanged under action by a unitary and if we act on $H_i', H_j'$ by a unitary, $C'$ changes equivalently, this paired property must exist for any simultaneously stoquastizable $H_i, H_j$. This holds for all pairs of Hamiltonians in $S$.
\end{proof}
This theorem provides a straightforward condition to rule out if $S$ is simultaneously stoquasticizable. However, the presence of paired eigenvalues does not guarantee simultaneous stoquasticity as: (a) stoquastic matrices must have negative, as well as real, off-diagonal elements; (b) it is possible for simultaneously non-stoquastic Hamiltonians to have a commutator with paired eigenvalues.

This condition also relates to the dynamical Lie algebra from quantum control theory~\cite{Jurdjevic1972,Ramakrishna1995} which for simultaneously real Hamiltonians (in any basis) neatly breaks up into a Cartan decomposition with every other layer of the dynamical Lie algebra (i.e. nested commutators with even numbers of our original Hamiltonians) corresponding to purely imaginary Hamiltonians (transformed into a given basis).  Therefore, finding a basis in which a set of Hamiltonians are simultaneously real, a necessary condition for simultaneous stoquasticization, is equivalent to identifying whether there is a Cartan decomposition of $\mathfrak{su}(d) = \mathfrak{p}\oplus \mathfrak{so}(d)$ where the set of Hamiltonians is contained in $\mathfrak{p}$.

\section{Bloch vector approach} We now reexpress the trace invariants in terms of Bloch vectors. This provides a geometric interpretation that neatly connects back to the intuition from the one qubit example given earlier, while highlighting a number of symmetries between words that are less clear in the alternative formalism. This approach will also allow us to prove another no-go result, from which the following theorem establishing the rareness of simultaneous stoquasticity immediately follows.

\begin{theorem}\label{thm:rarestoq}
For almost every $S$ with $m\geq 2$, $d\geq 3$, $S$ is not simultaneously stoquasticizable.
\end{theorem}
``Almost every'' is used in the technical sense that the set of simultaneously stoquastizable $S$ are measure zero.

With this goal in mind, let us consider expressing the trace invariants from Theorem~\ref{thm:simsim} in terms of Bloch vectors. For words of arbitrary length $|w|$, we have 
\begin{equation}\label{eq:gentrtobloch}
   \mathrm{Tr}[w(S)]=\mathrm{Tr}\left[\prod_{j=1}^{|w|}\sum_{\mu_j=1}^{d^2-1}b^{(w_j)}_{\mu_j} \hat\lambda_{\mu_{j}}\right],
\end{equation}
where we have denoted the $j$-th element of $w$ as $w_j$. 

Now consider evaluating Eq.~(\ref{eq:gentrtobloch}) explicitly for words of small length. By our assumption of tracelessness, the trace invariant for any $w(S)$ of length one is zero. A general trace invariant for $|w|=2$ is
\begin{align}
    \mathrm{Tr}[H_iH_j]=\mathrm{Tr}[(\vec b^{(i)}\cdot\hat{\vec\lambda})(\vec b^{(j)}\cdot\hat{\vec\lambda})]=2\vec b^{(i)}\cdot \vec b^{(j)},
\end{align}
where we used Eq.~(\ref{eq:basisproduct}) to evaluate the trace. Therefore, the lengths (from the $i=j$ case) and relative angles (from the $i\neq j$ case) of the Bloch vectors corresponding to pairs of Hamiltonians in $S$ are simultaneous trace invariants.

This result is intuitively satisfying. For a qubit, 
recalling that $SU(2)$ is homomorphic to $SO(3)$, this is precisely what we would expect to be invariant for a rigid collection of vectors being rotated simultaneously about the origin. We should also expect this to be the only constraint for a qubit. This expectation is validated by computing the trace invariant for words of length three:
\begin{align}\label{eq:w3}
     \mathrm{Tr}[H_iH_jH_k]=2d_{\mu\nu\xi}b^{(i)}_\mu b^{(j)}_\nu b^{(k)}_\xi=2(\vec b^{(i)}*\vec b^{(j)})\cdot \vec b^{(k)},
\end{align}
where we have introduced the star product, defined component-wise, using the symmetric structure constants, as,
\begin{align}\label{eq:star}
    (\vec b^{(i)}*\vec b^{(j)})_\xi=d_{\mu\nu\xi}b^{(i)}_\mu b^{(j)}_\nu.
\end{align}

For $\mathfrak{su}(2)$, the symmetric structure constants are all zero, so, as expected for a qubit, words of length greater than two provide no further constraints. 

Various properties of the star product are detailed in the Appendix. In particular, observe that the star product is not associative and that Eq.~(\ref{eq:w3}) is completely symmetric in the input word. Similar observations allow us to show that any trace invariant can be written as $\vec v \cdot \vec b^{(i)}$ for some $i\in[1,m]$, where $\vec v$ is any vector in the set $\mathcal{B}$ of all possible combinations of star products between Bloch vectors in $B$. That is $\mathcal{B}=\{ \vec b^{(j)}, \vec b^{(j)}*\vec b^{(k)}, (\vec b^{(j)}*\vec b^{(k)})*\vec b^{(l)}, \cdots \}$. This can be verified by direct computation, but we provide explicit proof in the Appendix. 

Given this formalism, we can pick a finite set of Bloch trace invariants using Eq.~(\ref{eq:lmax}) and then construct an equivalent decision problem to Eqs.~(\ref{eq:polysys1})-(\ref{eq:polysys3}) to test for simultaneous stoquasticity. The stoquasticity conditions in this context are $b'^{(i)}_j=0$ $j\in\mathcal{Y}$ and $b'^{(i)}_j\leq 0$ for $j\in\mathcal{X}$ for all $i$.

We also obtain the following no-go result.
\begin{theorem}\label{thm:nogobloch}
Let  $S$ be a set of Hermitian matrices with corresponding Bloch vectors $B=\{\vec b^{(1)}, \vec b^{(2)}, \cdots \vec b^{(m)}\}$. Let $\mathcal{B}$ be the set of all possible star products between elements of $B$. A necessary condition for $S$ to be simultaneously stoquasticizable is that $\mathrm{dim}(\mathrm{span}(\mathcal{B}) )\leq(d^2+d-1)/2$. 
\end{theorem}
\noindent\emph{Proof sketch.} Observe that for all $H_i\in\texttt{Stoq}$, $\vec b^{(i)}_j=0$ for $j\in\mathcal{Y}$. From the definition of the star product and the form of the non-zero symmetric structure constants of $\mathfrak{su}(d)$ \cite{bossion2021general}, one observes that if $S\in\texttt{Stoq}$ all vectors in $\mathcal{B}$ are also in this subspace. The dimension of $\mathcal{B}$ is invariant under unitary transformations, so this is a necessary condition for simultaneous stoquasticity. \hfill $\square $

Full details are provided in the Appendix. Most importantly, this result leads directly to Theorem~\ref{thm:rarestoq}, the proof of which we sketch below, again leaving the algebraic details to the Appendix.

\noindent\emph{Proof sketch of Theorem~\ref{thm:rarestoq}.} From similar analysis of the star products, we can prove that for almost every $S$, $\mathrm{dim}(\mathcal{B} )=d^2-1$. That is, $\mathcal{B}$ spans the full Bloch vector space for almost every $S$. Combining this result with Theorem~\ref{thm:nogobloch}, Theorem~\ref{thm:rarestoq} immediately follows. \hfill $\square $

\section{Conclusion and Outlook} 

Quantum annealing relies on the interaction of two non-commuting Hamiltonians, and there are clear connections between the power of that computation and the stoquasticity of the Hamiltonians.  Our results provide proof that a general quantum annealing procedure does not possess any basis in which it can be described completely stoquastically.  We know that classical computing can be described using simultaneously diagonal Hamiltonians, and the seeming power of non-stoquasticity speaks to the idea that quantum advantage might lie further past simultaneously stoquastic Hamiltonians.  More work is needed to determine how tightly quantum advantage is bound up merely with these notions of simultaneous stoquasticity and how much other factors, such as  locality of the simultaneous basis, play a role.

Furthermore, our results provide a definitive set of conditions for simultaneous stoquasticity, which are, as expected, difficult to calculate in practice given the computational complexity of this problem.  The commutator condition of Theorem~\ref{thm:nogo}, while enticing from its connections to simultaneous diagonalizability (see the Appendix for a further exploration of how these geometric ideas relate to simultaneous diagonalizability) and dynamical Lie algebras, provides only a necessary but not sufficient condition and then only on a simultaneous real basis, not specifically a simultaneous stoquastic basis. The other no-go result in Theorem~\ref{thm:nogobloch} suffers a similar flaw.

Further work is also possible by extending these results beyond just stoquasticity to the full class of Hamiltonians lacking sign problems. General sign problem free Hamiltonians take a Vanishing Geometeric Phase (VGP) form \cite{jarret2018hamiltonian, Hen2021}.  This form generalizes the notion of stoquastic Hamiltonians to all Hamiltonians generated from stoquastic Hamiltonians via diagonal unitary transformations. While this is a more general form that should be studied in the context of simultaneous transformations, it lacks linearity, meaning that linear combinations of VGP Hamiltonians are not necessarily VGP, hinting that simultaneous stoquasticizability is a more fundamental concept to consider for multiple Hamiltonians.

\acknowledgements  We thank Adam Ehrenberg, Luis Pedro Garc\'{i}a-Pintos, Alexey V. Gorshkov, Dominik Hangleiter, Michael Jarret, and Alexander F. Shaw for helpful discussions. J.B.~acknowledges support by the U.S.~Department of Energy, Office of Science, Office of Advanced Scientific Computing Research, Department of Energy Computational Science Graduate Fellowship (award No.~DE-SC0019323) and by the DoE ASCR Accelerated Research in
Quantum Computing program (award No.~DE-SC0020312), DoE QSA, NSF QLCI
(award No.~OMA-2120757), DoE ASCR Quantum Testbed Pathfinder program
(award No.~DE-SC0019040), U.S. Department of Energy Award
No.~DE-SC0019449, NSF PFCQC program, AFOSR, ARO MURI, AFOSR MURI, and
DARPA SAVaNT ADVENT.

During the preparation of this paper L.B. changed affiliation from QuICS/NIST to QuAIL/KBR.  L.B. is supported by the Prime Contract No. 80ARC020D0010 with the NASA Ames Research Center and is grateful for support from DARPA under IAA 8839 annex 128. The United States Government retains, and by accepting the article for publication, the publisher acknowledges that the United States Government retains, a nonexclusive, paid-up, irrevocable, worldwide license to publish or reproduce the published form of this work, or allow others to do so, for United States Government purposes.

\begin{appendix}

\section{Some properties of the star product}
Here we demonstrate a number of useful identities regarding the star product introduced in the main text. Some of these properties do not seem to be well documented in the literature on generalized Bloch vectors due to the focus on single density matrices. In these contexts, only star products between the same Bloch vector arise, which obscures some of the more general properties of the product. In particular, we emphasize that the product is non-associative.

We define the star product of two vectors $\vec a, \vec b\in\mathbb{R}^{d^2-1}$ component-wise as
\begin{equation}
    (\vec a * \vec b)_\xi = d_{\mu \nu \xi} a_\mu b_\nu,
\end{equation}
where $d_{\mu \nu \xi}$ are totally symmetric structure constants for $\mathfrak{su}(d)$. Observe that this product is basis-dependent due to the structure constants. 

The star product has the following properties which can be verified by explicit component-wise computation:
\begin{align*}
    &\vec a * \vec b = \vec b * \vec a & &(\text{commutative})\\
    &(\vec a * \vec b) * \vec c\not = \vec a * (\vec b * \vec c) & &(\text{non-associative})\\
    &\vec a * (\vec b + \vec c) = \vec a * \vec b + \vec a * \vec c & & (\text{distributive})\\
    &(\vec a * \vec b) \cdot \vec c =  (\vec b * \vec c) \cdot \vec a =  (\vec a * \vec c) \cdot \vec b.
\end{align*}

The last identity can be used to show that
\begin{align}
    (\vec a * \vec b) \cdot (\vec c * \vec d)&= ((\vec a * \vec b) * \vec c)\cdot \vec d \nonumber\\
    &= ((\vec a * \vec b) * \vec d)\cdot \vec c \nonumber \\
    &= ((\vec c * \vec d) * \vec a)\cdot \vec b \nonumber \\
     &= ((\vec c * \vec d) * \vec b)\cdot \vec a.
\end{align}
This result generalizes: the dot product of any combination of star products can be rearranged such that the dot product is just with a single vector at the end of the computation, provided one is careful with the non-associativity of the star product.

When it is not misleading, it can be convenient to adopt the convention that multiplication proceeds from left to right so we can drop the parenthesis and have, for instance, that
\begin{equation}
    (((\vec a * \vec b) * \vec c)* \vec d) = \vec a * \vec b * \vec c * \vec d.
\end{equation}

Finally, we introduce the notation that $\vec b^{*k}$ denotes the $k$-fold star product $\vec b * \vec b * \cdots \vec b$, such that $\vec b^{*1}=\vec b$, $\vec b^{*2}=\vec b*\vec b$ and so on.

\section{Proof of sufficient word length for unitary similarity of a pair of Hermitian matrices}\label{app:proofforpair}
\begin{theorem}\label{thm:densityproof}
Two Hermitian matrices $H, H'\in\mathbb{C}^{d\times d}$ are unitarily similar if and only if $\mathrm{Tr}[H^k]=\mathrm{Tr}[H'^k]$ for $k\in[1,d]$.
\end{theorem}
\begin{proof}
In general, from Theorem~\ref{thm:simsim}, it is necessary to check the trace conditions for all words $H^k, H'^k$. We show it is only necessary to check the first $d$ such words. That is, $\mathrm{Tr}[H^k]=\mathrm{Tr}[H'^k]$ for $k\in[1,d]$ implies $\mathrm{Tr}[H^k]=\mathrm{Tr}[H'^k]$ for $k>d$. 

Observe that traces are basis-independent so we may write the given set of equivalences for $k\in[1,d]$ as
\begin{equation}\label{eq:powersums}
    \mathrm{Tr}[H^k]=\mathrm{Tr}[H'^k]\implies\sum_{j=1}^d \lambda_j^k =  \sum_{j=1}^d \lambda_j'^k,
\end{equation}
where $\{\lambda_j\}$, $\{\lambda_j'\}$ are the eigenvalues of $H, H'$, respectively. Because $H, H'$ are Hermitian these are real. The sums in Eq.~(\ref{eq:powersums}) are known as the power sums $p_k(\lambda_1, \cdots, \lambda_d)$ and $p_k(\lambda'_1, \cdots, \lambda'_d)$. Via the Newton-Girard identities, one can explicitly write the first $d$ elementary symmetric polynomials $e_1, \cdots e_d$ in these eigenvalues in terms of the power sums $p_k$ for $k\in[1,d]$. By the equivalence of the power sums between the variables $\{\lambda_j\}$, $\{\lambda_j'\}$, the elementary symmetric polynomials in these two variables are also equivalent. 

The elementary symmetric polynomials then allow us to write the chain of equalities via a standard expansion of a polynomial in some variable $x$ with roots $\{\lambda_j\}$:
\begin{align}
    \prod_{j=1}^d &(x-\lambda_j)=\sum_{k=0}^n(-1)^ke_k(\lambda_1, \cdots,\lambda_n)x^{n-k}\nonumber \\
   &=\sum_{k=0}^n(-1)^ke_k(\lambda'_1, \cdots,\lambda'_n)x^{n-k}=\prod_{j=1}^d (x-\lambda'_j).
\end{align}
This implies that $\lambda_j=\lambda_j'$ for all $j\in[1,d]$. This in turn implies that $\mathrm{Tr}[H^k]=\mathrm{Tr}[H'^k]$ for $k>d$, proving the result.
\end{proof}
We remark that Theorem~\ref{thm:densityproof} is often stated without proof in the context of giving the independent trace invariants of density matrices $\rho\in\mathbb{C}^{d\times d}$ \cite{hioe1981nlevel, mahlerquantum, geometryquantum}. Given that Hermitian matrices have $d$ eigenvalues, the theorem is intuitively obvious, but we have not seen the explicit proof of this statement in the physics literature. 

We also provide an alternative statement of the theorem and a corresponding proof which makes use of the Bloch vector formalism for trace invariants. 

\begin{theorem}\label{thm:densityproofbloch}
Two traceless Hermitian matrices $H, H'\in\mathbb{C}^{d\times d}$ are unitarily similar if and only if $\vec{a}^{*k} \cdot \vec a=\vec b^{*k} \cdot \vec b$ for $k\in[1, d-1]$, where $\vec a, \vec b$ are the Bloch vectors corresponding to $H,H'$, respectively.
\end{theorem}
\begin{proof}
Consider the infinite set of vectors $\mathcal{A}=\{\vec a, \vec a* \vec a, \vec a *\vec a *\vec a, \cdots \}$. The invariants of $A$ under unitary transformations are $\vec v \cdot \vec a$ for any $\vec v\in \mathrm{span}(\mathcal{A})$. 

Any Hermitian matrix can be diagonalized via a unitary matrix and the invariants are unchanged under this transformation so assume that we have diagonalized $A$. Suppose we are using a generalized Gell-Mann basis, as described in the main text. Then the corresponding $\vec a$ has $a_j=0$ for $j\in\mathcal{X}, \mathcal{Y}$. That is, $\vec a$ is in a subspace $T$ of dimension $d-1$ spanned by $\hat{\vec\lambda}^{\mathrm{(diag)}}$, with corresponding indices $\mathcal{D}$.
We have that
\begin{align}
    (\vec a * \vec a)_k = d_{ijk} a_i a_j,
\end{align}
where the only non-zero terms in the sum correspond to non-zero $d_{ijk}$ with $i,j\in \mathcal D$. The only non-zero $d_{ijk}$ satisfying this condition have $k\in \mathcal{D}$. Therefore, $\vec a * \vec a\in T$. This is true for all $\vec v \in \mathcal A$. As $T$ is of dimension of $d-1$, only up to $d-1$ of the vectors in $\mathcal A$ are independent. Adding in the traceless condition, this makes for a maximum of $d$ independent invariants. Therefore, to determine the simultaneous similarity of $H, H'$ it is sufficient to check the equivalence of only the $d$ invariants in the theorem statement. 
\end{proof}
We remark if $H, H'$ are not traceless one can merely check if the matrices have the same trace and then apply the theorem above.

\section{Bound on the length of words}

This section seeks to prove the bound on $\ell_{max}$, the maximum length of word we need to check to capture all independent invariants, given in Eq.~(\ref{eq:lmax}).  This proof is just an application of the bound and construction provided in Ref.~\cite{gerasimova2013simultaneous}.

The idea of simultaneous unitary similarity is derived in Ref.~\cite{gerasimova2013simultaneous}, starting from whether two complex matrices are unitarily similar.  It is a known result that two $n\times n$ complex matrices, $A$ and $B$ (with $A^*$ and $B^*$ denoting their complex conjugates), are unitarily similar if and only if $\Tr [w(\{A,A^*\})] = \Tr [w(\{B,B^*\})]$ for every word $w(\{s,t\})$ of two non-commuting matrices whose length is less than or equal to 
\begin{equation}
    \label{eq:lmaxprime}
    \ell'_{\mathrm{max}}=\mathrm{min}\begin{cases}
    \ceil{\frac{n^2+2}{3}}\\
    n\sqrt{\frac{2n^2}{n-1}+\frac{1}{4}}+\frac{n}{2}-2.
    \end{cases}
\end{equation}
The $O(n^2)$ bound is due to Paz~\cite{paz1984application} and the asymptotically better $O(n^{3/2})$ bound is due to Pappacena~\cite{pappacena1997upper}. This result holds for arbitrary complex matrices, and as we discussed in the preceding section, the bound on considered words can be much tighter if we consider just the unitary similarity of two Hermitian matrices.  However, this more general bound is important in the context of the unitary similarity of sets of matrices.  In Ref.~\cite{gerasimova2013simultaneous}, the authors produce an encoding of sets of matrices into two larger matrices such that if those two larger matrices are unitarily similar then all the individual pairs of matrices from the two sets must be unitarily similar under the same transformation.  Furthermore, the word trace condition for unitary similarity of these larger matrices is equivalent to a word trace condition on all words of that same length compared between words made entirely of one of the sets and words made entirely from the other set.

Specifically, given two sets of matrices $S = \{s_1,s_2,\ldots,s_m\}$ and $S' = \{s'_1,s'_2,\ldots,s'_m\}$ where each matrix is of size $d\times d$, we can encode these sets into matrices $A$ and $B$.  $A$ will be a block matrix constructed of $d\times d$ matrices, and the diagonal and all blocks below it are zero.  Immediately above the diagonal are $d\times d$ identity matrix blocks, and into the remaining portion of the upper triangular portion, we slot the matrices from $S$ into the blocks.  In order to do this, we need $m$ spaces remaining in the upper triangle.  If the $A$ matrix is $c$ blocks long, then there will be $(c^2-3c+2)/2$ spaces for matrices from our set $S$.  Thus, we must choose $c$ such that $(c^2-3c+2)/2\geq m$.  Any unused blocks are set to zero.  The $B$ matrix is constructed similarly except using $S'$ instead of $S$.  These matrices will both be of size $n=cd$.

In Ref.~\cite{gerasimova2013simultaneous} it is proven that the matrices $A$ and $B$ are unitarily similar if and only if the sets $S$ and $S'$ are unitarily similar.  Furthermore, the trace word conditions on $A$ and $B$ being unitarily similar is equivalent to the condition that all words have equal traces between sets $S$ and $S'$ with the lengths of these necessary words being bounded by the same length as the words necessary to check unitary similarity of $A$ and $B$.  Therefore, it is sufficient to check words up to length based off Eq.~(\ref{eq:lmaxprime}) with $n=cd$, recovering Eq.~(\ref{eq:lmax}).

\section{Trace invariants in terms of Bloch vectors}
In this section, we demonstrate the claim from the main text that all trace invariants under simultaneous unitary transformations can be expressed as linear combinations of invariants $\vec v\cdot \vec b^{(j)}$ for $\vec v\in\mathcal{B}$ and $\vec b^{(j)}$ a Bloch vector corresponding to the Hamiltonian $H_j$. Recall, we define $\mathcal{B}$ to be the set of all possible star products between Bloch vectors in $B$.

Furthermore, recall that for a trace invariant for a word of arbitrary length we have
\begin{equation}\label{eq:gentrtobloch_app}
    \mathrm{Tr}[w(S)]=\mathrm{Tr}\left[\prod_{j=1}^{|w|}\sum_{\mu_j=1}^{d^2-1}b^{(w_j)}_{\mu_j} \hat\lambda_{\mu_{j}}\right],
\end{equation}
where we have denoted the $j$-th element of $w$ as $w_j$. We can flip the order of the product and the sum and write
\begin{equation}\label{eq:gentrtobloch_app2}
    \mathrm{Tr}[w(S)]=\sum_{\{\mu_1\cdots \mu_{|w|}\}}\mathrm{Tr}\left[\prod_{j=1}^{|w|}b^{(w_j)}_{\mu_j} \hat\lambda_{\mu_{j}}\right],
\end{equation}
where the sum is over all ordered sets (with replacement) of indices $\in [1, d^2-1]$. We now make use of Eq.~(\ref{eq:basisproduct}) to evaluate the products of basis elements of $\mathfrak{su}(d)$. Due the sum over all ordered sets the antisymmetric terms in each product of $\hat\lambda_{\mu_j}$ cancel and we are left to consider the identity terms and the terms with symmetric structure constants. As the basis elements are all traceless, what ultimately survives the trace once we fully expand out all products are the terms proportional to identity. 

Under an expansion of the products and evalaution of the trace the term with the most symmetric structure constants is proportional to
\begin{align}
    (d_{\mu_1\mu_2\nu_1} d_{\mu_3\nu_1\nu_2}d_{\mu_4\nu_2\nu_3}&\cdots d_{\mu_{|w|-1}\nu_{|w|-3}\nu_{|w|-2}}\delta_{\nu_{|w|-2}\mu_{|w|}})\nonumber \\
    &\times (b^{(w_j)}_{\mu_1}\cdots b^{(w_{|w|})}_{\mu_{|w|}}),
\end{align}
where we use the convention of summing over repeated indices. After staring at the proliferation of indices, one observes that this term can be compactly written as
\begin{align}
   (((\vec b^{(w_1)}*\vec b^{(w_2)})*\vec b^{(w_3)})*\cdots*\vec b^{(w_{|w|-1})})\cdot\vec b^{(w_{|w|})}
\end{align}
All other terms in the expansion of the products of basis elements of $\mathfrak{su}(d)$ consist of fewer symmetric structure constants and more Kronecker deltas. These other terms amount to dot products between terms similar to this one but with smaller word length. Therefore, this term is the only one that is not dependent on invariants established from smaller length words. By the commutativity of the star product we may permute any products we like, provided we respect the lack of associativity. From this, and the fact that all words yield trace invariants, we establish the intended claim.

This result implies a nice geometric interpretation of the trace invariants. In particular, the relative angles between all vectors in $\mathcal{B}$ are the invariants under unitary transformations. This is a manifestation of the fact that $SU(d)\subset SO(d^2-1)$. Due to the asymmetry in Bloch space of the symmetric structure constants, the set of star products $\mathcal{B}$ are not, in general, rotationally invariant---the rotations where this infinite set of vectors \emph{do} rotate rigidly picks out the rotations corresponding to $SU(d)$.

\section{Simultaneous stoquasticizability is rare}
For all the proofs in this section it will be necessary to identify the non-zero symmetric structure constants of $\mathfrak{su}(d)$ in the generalized Gell-Mann basis. We take the explicit form of these from Ref.~\cite{bossion2021general}, with some slight differences in indexing to account for our conventions differing from those used by those authors \footnote{In particular, relative to the conventions of Ref.~\cite{bossion2021general}, we do the following: (a) we index our diagonal basis elements of the generalized Gell-Mann basis $[1,d-1]$ as opposed to to $[2, d]$; (b) we index the symmetric and anti-symmetric basis elements in increasing, rather than decreasing order. For instance, for $\hat\lambda_{ij}^{(x)}$ we have $i<j$, not $i>j$; (c) our basis elements are a factor of two larger and therefore our trace orthonormality condition is four times larger. That is, they have $\mathrm{Tr}(\hat\lambda_i\hat\lambda_j)=\delta_{ij}/2$, whereas we have $\mathrm{Tr}(\hat\lambda_i\hat\lambda_j)=2\delta_{ij}$. Despite this, the normalization of the structure constants agree}. We identify these symmetric structure constants based on whether the indices correspond to symmetric (Eq.~\ref{eq:gmx}), skew-symmetric (Eq.~\ref{eq:gmy}), or diagonal (Eq.~\ref{eq:gmd}) basis elements. In particular, we give a one-to-one mapping between indices $i\in[1,d^2-1]$ and indices $\mathcal{X}_{jk}, \mathcal{Y}_{jk}$ and $\mathcal{D}_j$ corresponding to the sets of symmetric, skew-symmetric and diagonal basis element, respectively, as follows:
\begin{align}
    \mathcal{X}_{jk}&=k^2+2(j-k)-1\\
    \mathcal{Y}_{jk}&=k^2+2(j-k)\\
    \mathcal{D}_{j}&=j(j+2),
\end{align}
where $1\leq j <k \leq d$. Let $\mathcal{X}=\{\mathcal{X}_{jk}\}$, $\mathcal{Y}=\{\mathcal{Y}_{jk}\}$, and  $\mathcal{D}=\{\mathcal{D}_{j}\}$. Given such an identification, we have the following non-zero symmetric structure constants:
\begin{align}\label{eq:structureconstants}
    &d_{\mathcal{X}_{jk}\mathcal{X}_{lj}\mathcal{X}_{lk}}=d_{\mathcal{X}_{jk}\mathcal{Y}_{lj}\mathcal{Y}_{lk}}=d_{\mathcal{X}_{jk}\mathcal{Y}_{kl}\mathcal{Y}_{jl}}=-d_{\mathcal{X}_{jk}\mathcal{Y}_{jl}\mathcal{Y}_{lk}}=\frac{1}{2} \nonumber\\
    &d_{\mathcal{X}_{jk}\mathcal{X}_{jk}\mathcal{D}_{j-1}}=d_{\mathcal{Y}_{jk}\mathcal{Y}_{jk}\mathcal{D}_{j-1}}=-\sqrt{\frac{j-1}{2j}} \nonumber\\
    &d_{\mathcal{X}_{jk}\mathcal{X}_{jk}\mathcal{D}_{l-1}}=d_{\mathcal{Y}_{jk}\mathcal{Y}_{jk}\mathcal{D}_{l-1}}=\sqrt{\frac{1}{2l(l-1)}}, \quad j< l<k \nonumber\\
     &d_{\mathcal{X}_{jk}\mathcal{X}_{jk}\mathcal{D}_{k-1}}=d_{\mathcal{Y}_{jk}\mathcal{Y}_{jk}\mathcal{D}_{k-1}}=\frac{2-k}{\sqrt{2k(k-1)}} \nonumber\\
    &d_{\mathcal{X}_{jk}\mathcal{X}_{jk}\mathcal{D}_{l-1}}=d_{\mathcal{Y}_{jk}\mathcal{Y}_{jk}\mathcal{D}_{l-1}}=\sqrt{\frac{2}{l(l-1)}}, \quad k < l \nonumber\\
    &d_{\mathcal{D}_{j-1}\mathcal{D}_{k-1}\mathcal{D}_{k-1}}=\sqrt{\frac{2}{j(j-1)}}, \quad  k < j \nonumber\\
    &d_{\mathcal{D}_{j-1}\mathcal{D}_{j-1}\mathcal{D}_{j-1}}=(2-j)\sqrt{\frac{2}{j(j-1)}}
\end{align}
We shall find for the following proofs that it is sufficient to observe which symmetric structure constants are non-zero, but for explicit application of these theorems these analytic expressions would be convenient. 

As implied by the structure constants being considered, we will make use of the generalized Gell-Mann basis throughout these proofs, but we observe that all trace invariants are basis independent. Consequently, the dimensions of spaces spanned by the star products of Bloch vectors are also basis independent quantities. 

We begin with a simple theorem describing the maximum size of the set of star products arising from simultaneously stoquastic matrices.

\begin{theorem}[Theorem 4 from the main text]\label{thm:stoqdim}
Let  $S=\{H_1, \cdots, H_m\}$ be a set of Hermitian matrices with corresponding Bloch vectors $B=\{\vec b_1, \vec b_2, \cdots \vec b_m\}$. Let $\mathcal{B}$ be the set of all possible star products between elements of $B$. A necessary condition for the elements of $S$ to be simultaneously stoquastizable is that $\mathrm{dim}(\mathrm{span}(\mathcal{B}) ) )\leq(d^2+d-1)/2$.
\end{theorem}
\begin{proof}
Suppose $H\in\texttt{Stoq}$ for all $H\in S$. Observe from the list of non-zero symmetric structure constants in Eq.~(\ref{eq:structureconstants}) that any star products between elements of $B$ necessarily have all components $j\in \mathcal{Y}$ equal to zero. That is, any star product between vectors in the $\mathcal{X}, \mathcal{D}$ subspace remain in that subspace. Therefore, $\mathrm{dim}(\mathrm{span}(\mathcal{B}) )\leq(d^2+d-1)/2$, where $(d^2+d-1)/2$ is the dimension of this subspace. The dimension of this subspace is preserved under unitary transformations as the relative angles between all elements of $\mathcal{B}$ are preserved under unitary transformations. This establishes the necessary condition for simultaneous stoquasticity in the theorem statement.
\end{proof}

From here, our goal will be to look at the space of star products arising from general Bloch vectors and show that this space is generally much larger than the space spanned by simultaneously stoquastic Bloch vectors.
 We begin with a useful lemma.
\begin{lemma}\label{lemma:diagonal}
Let $H$ be a traceless Hermitian matrix with corresponding Bloch vector $\vec b$. Let $\mathcal{B}=\{\vec b^{* k}\, |\,k\in \mathbb{Z}^+\}$ be the (infinite) set of all possible star products of $\vec b$. Then $\mathrm{dim}(\mathrm{span}(\mathcal{B}))\leq d-1$
\end{lemma}
\begin{proof}
As $\mathrm{dim}(\mathrm{span}(\mathcal{B}))$ is a basis-independent property, assume that $H$ is diagonal without loss of generality. Therefore, the generalized Gell-Mann basis $\vec b_j=0$ for all $j\in\{\mathcal{X}, \mathcal{Y}\}$. From the definition of the star product and the form of the structure constants in Eq.~(\ref{eq:structureconstants}) observe that the only non-zero components of the star $\vec b * \vec b$ are those with index $j\in \mathcal{D}$. This holds for the star product of any pair of vectors $\vec u,\vec v$ with $u_j, v_j=0$ for all $j\in\{\mathcal{X}, \mathcal{Y}\}$. Therefore, all $\vec v\in\mathcal{B}$ are necessarily contained in the $d-1$ dimensional subspace with indices $\mathcal{D}$.
\end{proof}
If $H$ in the above lemma is \emph{generic}---i.e. has non-degenerate eigenvalues---$\mathrm{dim}(\mathrm{span}(\mathcal{B}))=d-1$. For non-generic Hamiltonians, there exist additional symmetries which leads to dependence between the elements of $\mathcal{B}$ and, consequently, a reduction in the dimension of $\mathrm{span}(\mathcal{B})$. For instance, in the extreme case of $H$ corresponding to a pure state density matrix $\rho$---which has zero as a $(d-1)$-fold degenerate eigenvalue---we have that $\vec b*\vec b=\vec b$, implying that $\mathrm{dim}(\mathrm{span}(\mathcal{B}))=1$ \cite{byrd2003characterization}.

We will also make use of the following lemma, proving a linear algebra fact that will be useful later on.
\begin{lemma}\label{lemma:spanxy}
Let $R\in\mathbb{R}^{2n\times 2n}$ be a diagonal matrix with  all non-zero, 2-fold degenerate matrix elements such that $R_{jj}=R_{kk}$ for $k=2j$ and $R_{jj}\neq R_{kk}$ otherwise. Let $\vec u,\vec v\in\mathbb{R}^{2n}$ be vectors with all unique elements such that $u_{j^*}=0$, $v_{2j^*}=0$ for some particular $j^*$. Then the vectors in the set $\{R^k \vec u, R^k \vec v\}$ for $k\in \mathbb{Z}^+$ span $\mathbb{R}^{2n}$.
\end{lemma}
\begin{proof}
With the exception of the zero vector, $\mathrm{span}(\{R^k\vec u\})$ is completely disjoint from $\mathrm{span}(\{R^k \vec v\})$ from the uniqueness conditions on $R, \vec u, \vec v$ and the fact that $u_{j^*}=0, v_{2j^*}=0$. In particular, there exists no non-trivial $q_k, r_k\in\mathbb{R}$ such that $\sum_k q_k R^k \vec u=\sum_k r_k R^k \vec v \implies \sum_k q_k \vec u=\sum_k r_k \vec v $, as this would require $\sum_k q_k u_{2j^*}=0$ and $\sum_k r_k v_{j^*}=0$, implying $\sum_k q_k=\sum_k r_k=0$.

Now consider the span of $\{R^k\vec u\}$. As the components of $\vec u$ are unique and $R$ has $n$ unique components, these vectors will span a space of dimension $n$. This is because for $R^k\vec u$ to be linearily dependent on $\{R^l \vec u\}$ for $l<k$ requires that there not exist constants $c_l\in\mathbb{R}$ such $R^k\vec u=\sum_{l=1}^{k-1} c_l R^l \vec u$, which, by the uniqueness of the components of $\vec u$ implies we require $R^k=\sum_{l=1}^{k-1} c_l R^l$ for dependence. Such constants only exist for $k\geq n$ due to the uniqueness conditions on $R$. 

An identical argument holds for the span of $\{R^k\vec u\}$. As the two spans are completely disjoint, together they span the full vector space of dimension $2n$.
\end{proof}

Armed with the preceding two lemmas, we are now prepared to prove the following important theorem.. Here, ``almost every'' is used in a technical sense, in that the set of possibilities not obeying the given condition are (Lebesgue) measure zero.
\begin{theorem}\label{thm:stardim}
Let  $S=\{H_1, \cdots, H_m\}$ be a set of Hermitian matrices with corresponding Bloch vectors $B=\{\vec b^{(1)}, \vec b^{(2)}, \cdots \vec b^{(m)}\}$. Let $\mathcal{B}$ be the (infinite) set of all possible star products between elements of $B$. For almost every $S$ with $m\geq 2, d\geq 3$, $\mathrm{dim}(\mathrm{span}(\mathcal{B}) )=d^2-1$. That is, $\mathcal{B}$ spans the full Bloch vector space for almost every $S$.
\end{theorem}
\begin{proof}
It is sufficient to consider $m=2$ as the dimension of the space spanned by this subset is less than or equal to that of the full set $\mathcal{B}$. Without loss of generality, assume that $H_1$ is diagonal so that in the generalized Gell-Mann basis $\vec b^{(1)}_j=0$ for all $j\in\{\mathcal{X}, \mathcal{Y}\}$. By Lemma~\ref{lemma:diagonal} and the discussion that follows it, the star products $(\vec b^{(1)})^{*k}$ for $k\in\mathbb{Z}^+$ span the $d-1$ dimensional space corresponding to indices $j\in\mathcal{D}$ for any $H_1$ with non-degenerate eigenvalues. Hermitian matrices with degenerate eigenvalues are measure zero in the space of traceless, Hermitian matrices \footnote{In particular, Hermitian matrices with repeated eigenvalues have codimension three (c.f. \cite{arnold1995remarks})} and, therefore, almost every $H_1$ is such that  $(\vec b^{(1)})^{*k}$ for $k\in\mathbb{Z}^+$ span the full $d-1$ dimensional space corresponding to indices $j\in\mathcal{D}$.

As these products of the form $(\vec b^{(1)})^{*k}$ span the $d-1$ dimensional space corresponding to indices $j\in\mathcal{D}$, we now seek to show that other elements of $\mathcal{B}$ span the remaining $\mathcal{X}, \mathcal{Y}$ components for almost every $\vec b^{(2)}$. To this end, consider the vectors $\vec b^{(2)}$ and $\vec b^{(2)}* \vec b^{(2)}$. 
Then, consider just the components of these vectors in the $\mathcal{X}, \mathcal{Y}$ subspace. Call these restricted vectors $\vec u, \vec v$. That is, $u_j=0$ for $j\in \mathcal{D}$ and $u_j=\vec b^{(2)}_j$ for $j\in\{\mathcal{X}, \mathcal{Y}\}$, and similarly for $\vec v$. For almost all $\vec b^{(2)}$, the corresponding $\vec u, \vec v$ are linearily independent. Therefore, we can construct via linear combinations two new vectors $\vec u', \vec v'$ such that $u'_{\mathcal{X}_{12}}=0, v'_{\mathcal{Y}_{12}}=0$. Such linear combinations are, by definition, in the span of $\mathcal{B}$.

Now, consider acting from the right on this linear combinations by star products of the form $(\vec b^{(1)})^{*k}$---i.e. consider elements in the span of $\mathcal{B}$ of the form $((\vec u'*\vec b^{(1)})*\vec b^{(1)})*\cdots$. 
From the definition of the star product, the fact that $\vec b^{(1)}_j=0$ for all $j\in\{\mathcal{X}, \mathcal{Y}\}$, and the form of the structure constants in Eq.~(\ref{eq:structureconstants}) we observe that such star products by $(\vec b^{(1)})^{*k}$ (from the right) act to scale the $\mathcal{X}, \mathcal{Y}$ components of $\vec u'$, $\vec v'$ by symmetric structure constant-dependent factors. We can write this scaling behavior as
\begin{equation}\label{eq:uscale}
    \vec u^{(k)}=R^k\vec u',
\end{equation}
where $R$ is a diagonal matrix with components given by
\begin{align}
    R_{jj}=\begin{cases}
    \sum_{i\in \mathcal{D}}d_{jji}b^{(1)}_i, & j\in\{\mathcal{X}, \mathcal{Y}\}\\
    0, & j\in\mathcal{D}
    \end{cases}
\end{align}
which comes from identifying the non-zero terms in the corresponding star product. There is an identical equation to Eq.~(\ref{eq:uscale}) for $\vec v'$. Importantly, we observe that $R_{jj}=R_{kk}$ for $j=\mathcal{X}_{lm}$, $k=\mathcal{Y}_{lm}$. This can be determined by observing in Eq.~(\ref{eq:structureconstants}) that $d_{\mathcal{X}_{lm}\mathcal{X}_{lm}j}=d_{\mathcal{Y}_{lm}\mathcal{Y}_{lm}j}$ for all $j\in\mathcal{D}$. Otherwise, for almost every $H_1$, $R_{jj}\neq R_{kk}$ for $j,k\in\{\mathcal{X},\mathcal{Y}\}$. From Lemma~\ref{lemma:spanxy}, $\{\vec u'^{(k)}, \vec v'^{(k)}\}$ span the $\mathcal{X}, \mathcal{Y}$ subspace of Bloch vector space for almost all $H_1, H_2$.
\end{proof}

Finally, we arrive at the following theorem which establishes that simultaneous stoquasticity is rare.
\begin{theorem}[Theorem 3 from main text]
Let  $S=\{H_1, \cdots, H_m\}$ be a set of Hermitian matrices. For almost every $S$ with $m\geq 2$, $d\geq 3$, $S$ is not simultaneously stoquasticizable.
\end{theorem}
\begin{proof}
This is an immediate consequence of Theorem~\ref{thm:stoqdim} and Theorem~\ref{thm:stardim}.
\end{proof}

\section{Results on simultaneous diagonalizability}
Similar results to those in the previous section can also be obtained for the problem of determining the simultaneous diagonalizability of a set of Hermitian matrices. This problem is, of course, well-known to be related to the problem of mutual compatibility of observables. Therefore, while taking an approach similar to that in our paper for this problem is largely overcomplicated compared to applying the simple condition that a set of Hermitian matrices are simultaneously diagonalizable if and only if they all commute, it is useful to compare the formalism established in this work to such conditions. 

In particular, as discussed in the main text, we expect that a deeper understanding of how our conditions relate to commutator conditions will enable connections to the dynamical Lie algebra of quantum optimal control theory. The case of simultaneous diagonalizability, with its well-known commutation condition provides a possible route forward.

We have the following theorem, which is analogous to Theorem~\ref{thm:stoqdim} (Theorem~\ref{thm:nogobloch} from the main text). The logic is also similar to that of Lemma~\ref{lemma:diagonal}, but extended to multiple matrices.

\begin{theorem}\label{thm:diagdim}
Let  $S=\{H_1, \cdots, H_m\}$ be a set of Hermitian matrices with corresponding Bloch vectors $B=\{\vec b^{(1)}, \vec b^{(2)}, \cdots \vec b^{(m)}\}$. Let $\mathcal{B}$ be the (infinite) set of all possible star products between elements of $B$ . A necessary condition for the elements of $S$ to be simultaneously diagonalizable is that $\mathrm{dim}(\mathrm{span}(\mathcal{B}) )\leq d-1$.
\end{theorem}
\begin{proof}
The proof is identical to that of Theorem~\ref{thm:stoqdim}, except, in this case, the star products of Bloch vectors corresponding to simultaneously diagonal Hermitian matrices are confined to the $d-1$ dimensional subspace of Bloch vector space with all components $j\in\mathcal{X}, \mathcal{Y}$ equal to zero. Therefore, a set of simultaneously diagonal Hermitian matrices has $\mathrm{dim}(\mathrm{span}(\mathcal{B}) )\leq d-1$. Again, the dimension of this subspace is preserved under unitary transformations, proving the result. 
\end{proof}
Combined with Theorem~\ref{thm:stardim} we have the following corollary as an immediate consequence. 
\begin{corollary}
Let  $S=\{H_1, \cdots, H_m\}$ be a set of Hermitian matrices. For almost every $S$ with $m\geq 2$, $d\geq 3$, $S$ is not simultaneously diagonalizable.
\end{corollary}

One also expects that amongst simultaneously stoquastic Hamiltonians simultaneously diagonalizable Hamiltonians are vanishingly rare. Proving this would follow a similar line of reasoning to that in the previous section. In particular, it would be sufficient to prove that for almost every set of stoquastic Hamiltonians the inequality in Theorem~\ref{thm:stoqdim} is tight. Unfortunately, the approach in Theorem~\ref{thm:stardim} doesn't immediately apply here since we can't diagonalize one of the matrices in a set of stoquastic Hamiltonians and keep all Hamiltonians stoquastic.

\end{appendix}

\bibliography{main.bib}

\end{document}